\documentclass{article} % For LaTeX2e
\usepackage{natbib}
\usepackage{graphicx}
\usepackage{booktabs}
\usepackage{amsmath}
\usepackage{amsthm}
\newtheorem{theorem}{Theorem}
\usepackage{caption}
\usepackage{wrapfig}
\usepackage{amsmath,amsfonts,bm}

\def\eqref#1{equation~\ref{#1}}
\def\1{\bm{1}}

\DeclareMathAlphabet{\mathsfit}{\encodingdefault}{\sfdefault}{m}{sl}
\SetMathAlphabet{\mathsfit}{bold}{\encodingdefault}{\sfdefault}{bx}{n}

\usepackage{hyperref}
\usepackage{url}
\usepackage[T1]{fontenc}
\usepackage[letterpaper, total={5.5in, 9in}]{geometry}

\title{Beyond Discrimination: Calibrated Geoprior Fusion for Bioacoustic Monitoring}

\author{
Neha Sajja\textsuperscript{1,2},
Bart van Merri\"{e}nboer\textsuperscript{1},  \\
Burcu Karagol Ayan\textsuperscript{1},
Tom Denton\textsuperscript{1} \\  % Adds a little breathing room
\textsuperscript{1}Google DeepMind \\
\textsuperscript{2}Harvard University
\\
\texttt{neha\_sajja@mde.harvard.edu}, \\ 
\texttt{\{bartvm,burcuka,tomdenton\}@google.com}}
\begin{document}

\maketitle

\begin{abstract}
Modern bioacoustic foundation models like Perch and BirdNET can identify species with high discriminative accuracy, yet their confidence scores are often uncalibrated and difficult to interpret as probabilities of real-world occurrence. This limits their use for ecological inference beyond threshold-based detection. We leverage a global annotated acoustic dataset (WABAD) to produce calibration priors for an acoustic model, optionally incorporating species-level information. 
We introduce new methods of fusing the acoustic predictions with geopriors, which empirically improves calibration while preserving discrimination. Together, these results suggest a path to simpler and more reliable acoustic monitoring for broad biodiversity.
\end{abstract}

\section{Introduction}
\label{intro}
Bioacoustics is increasingly used at large for scale passive acoustic monitoring (PAM) for ecological insight, but easily available classifiers, such as Perch v2~\citep{vanmerrienboer2025perch} or BirdNET~\citep{kahl2021BirdNET}, do not produce calibrated outputs, despite strong classification performance~\citep{wood2024guidelines, schwinger2025uncertainty}. A well-calibrated model is desirable both for choosing precision-based thresholds~\citep{august2022emerging} or estimating call density ~\citep{navine2024thresholds} as an indicator for species abundance. Model calibration is thus a helpful property for reliable analysis from bioacoustic data.

Calibrating a model typically requires expert validation effort on data gathered within a specific project. However, these experts are rare, and the required validation effort is onerous, especially when trying to monitor a large number of species across a wide variety of sites.

Simple presence or absence of the target species has a strong impact on the interpretation of model scores: If a species is geographically unlikely, it is unlikely to appear in an audio recording. Previous work has shown that species \textit{geopriors} can improve discriminative performance of classifiers, both for images~\citep{macaodha2019presence} and acoustics~\citep{jeantet2023}. However, fusion in previous work is achieved through simple multiplication of the probabilities from the perception model and the geoprior model. Because these are both bounded probability estimates, the product compresses towards zero and can lead to under-confident predictions, even as the discrimination of relative probabilities improves.

In this work, we attempt to obtain calibrated outputs from existing acoustic models for broad biodiversity monitoring with no additional validation work. To do this, we evaluate whether calibration parameters learned on a global acoustic dataset can improve an acoustic model's calibration when transferred to unseen datasets. We test two different ways of obtaining these calibration parameters from the global acoustic dataset and evaluate the methods using using two different acoustic foundation models across two unseen datasets. We then evaluate three distinct methods for fusing these pre-calibrated values with geopriors to determine if fusion improves discrimination and calibration and which method fusion method generalizes the best across datasets, and acoustic models. 

Towards this end, we make the following contributions: (1) We demonstrate that calibration parameters learned from a large global acoustic dataset effectively transfers to new datasets; (2) We show that careful \textit{fusion} of these pre-calibrated scores with geopriors can further improve discrimination and calibration.  We also demonstrate that our findings extend across distinct datasets and multiple bioacoustic foundation models.

\section{Related Work}
\label{related_work}

\paragraph{Deep Learning in Bioacoustics}
Foundational bioacoustic models, such as Perch v2 and BirdNET, have achieved state-of-the-art accuracy in identifying thousands of global species. However the confidence scores from these models do not reflect real-world likelihoods and thus practitioners rely on often arbitrarily chosen thresholds to determine species presence, which introduces false positive and false negative biases \citep{knight2017recommendations, cole2022occupancy}.

\paragraph{Calibration of Neural Networks and the Labeling Bottleneck}
Deep audio classifiers, like other deep neural networks \citep{guo2017calibration, ye2022uncertainty} require post-hoc calibration efforts (e.g., Platt scaling) to align confidence with actual correctness \citep{ platt1999probabilistic}. In multi-label bird sound classification, foundation models exhibit highly variable calibration due to distribution shifts \citep{schwinger2025uncertainty}. 
Crucially, calibration techniques demand additional validation data. In ecology, this manual verification process is labor-intensive and outpaces the capacity of domain experts \citep{sugai2019terrestrial, barre2019accounting}. Consequently, validation creates a severe bottleneck that reduces the scalability and reliability of automated monitoring.

\paragraph{Geographic Context and Geographic Priors}
Species Distribution Models (SDMs) estimate the likelihood of species presence in an area, providing continuous estimations of habitat suitability~\citep{elith2009species}. Geofencing selects a list of likely species for a given location. Multiplicative fusion of geopriors with classifier predictions can substantially improve discrimination for camera trap images~\citep{macaodha2019presence} and PAM data~\citep{jeantet2023}.

\section{Methods}
\label{sec:methods}
\begin{figure}
  \centering
  \vspace{-3mm}
  \includegraphics[trim=5mm 5mm 5mm 5mm, clip, width=\columnwidth]{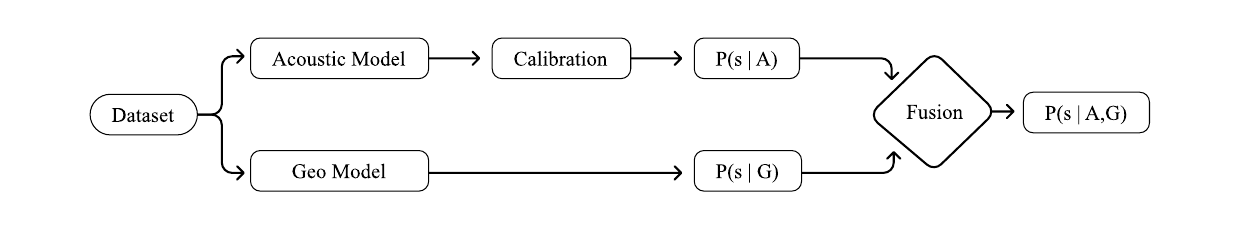}
  \vspace{-2mm}
  \caption{Methodology overview.}
  \label{fig:method_overview}
  \vspace{-3mm}
\end{figure}

\subsection{Models}
\label{sec:models}

The acoustic classification models evaluated are Perch~v2 \citep{vanmerrienboer2025perch} and BirdNET~v2.4 \citep{kahl2021BirdNET}. For geopriors, we use the BirdNET geographic meta-model. This is a temporal geoprior, accepting a latitude/longitude pair along with the week-of-year.

\subsection{Datasets and Ecological Environments}
\label{sec:datasets}

Fully annotated bioacoustic datasets are both rare and difficult to produce, but provide valuable ground-truth for multi-species calibration. Annotated datasets with precise multi-site geolocation information are even more rare. 

\textbf{Doohan A2O (A: Eastern Australia)} is a fully-annotated subset of the Australian Acoustic Observatory (A2O)~\citep{a2o}. The dataset includes 39 A2O sites, yielding 14,179 time windows across 157 species and 22,668 positive annotations. This serves as our development dataset~\footnote{We expect this dataset to be made publicly available soon.}.

\textbf{PANWAR (P: North America)} contains annotated PAM data from the northeastern United States from 104 sites, yielding 155,280 time windows across 104 species and 183,398 positive annotations. We use PANWAR as a second evaluation dataset to test the generalization of our leading methods~\citep{panwar2025largescale}.

\textbf{WABAD (Global)} contains 72 sites, yielding 100,469 time windows across 1,192 species. WABAD is used exclusively to learn pre-calibration parameters ~\citep{perezgranados2026wabad}.

\subsection{Mathematical Notation: Calibration and Fusion}
\label{sec:notation}

We compare methods for \textit{fusing} predictions from acoustic detection models with geographic priors into a single posterior probability. 
Let $z_{\mathrm{audio}}$ represent the raw logits from an acoustic model. $P(s)$ denotes the probability of presence of a target species in a specific spatiotemporal window, and $P(\bar{s})$ denotes its absence. We define $f(\cdot)$ as a chosen pre-calibration function that maps raw logits into acoustic probabilities, $P(s \mid A) \sim f(z_{\mathrm{audio}})$, where $A$ denotes the acoustic evidence. 

The geoprior model provides a corresponding probability, $P(s \mid G)$ where $G$ denotes the geographic evidence. Finally, we define $\phi(\cdot)$ as the fusion operator that combines the two probabilities into a candidate joint posterior probability,
\begin{equation}
    P_{\mathrm{fused}} = P(s|A,G)
    \simeq
    \phi\left(P(s \mid A), P(s \mid G)\right).
    \label{eq:fusion_operator}
\end{equation}

\subsection{Acoustic Models and Pre-Calibration Methods}
\label{sec:precalibration}

Acoustic models produce uncalibrated logits $z$. To map these logits to continuous probabilities, we apply the standard Platt scaling method \citep{platt1999probabilistic}, where $P(s|z) \sim \sigma(m_s z + b_s)$. The slope $m_s$ and intercept $b_s$ parameters for each species $s$ are found by logistic regression, minimizing the cross-entropy loss between the acoustic classifier's raw logits and the ground-truth labels. All Platt scaling in this paper is computed with the sklearn LogisticRegression library, with the default inverse regularization strength $C=1.0$. 
\begin{equation} P(s \mid A) \sim \sigma(m_s z + b_s) = \frac{1}{1 + \exp(-(m_s z + b_s))} \label{eq:Platt_probability} \end{equation}

\textbf{Baseline Raw Probabilities} are obtained by passing raw acoustic logits directly through a standard sigmoid $P_{\mathrm{RAW}}(s \mid A) = \sigma(z)$ and use these values as our uncalibrated probabilities. Note that Perch 2.0 was trained with softmax activation, so this sigmoid activation yields very poor calibration.

\textbf{Mean Likelihood Platt Scaling (MLPS)}:
Each species with at least one positive example in the WABAD dataset is calibrated with Platt scaling, producing a collection of 1,147 Platt parameters, $(m_s, b_s)$, illustrated in ~\autoref{fig:coef_variation}. For MLPS, we take the average slope $m_G$ and intercept $b_G$ across all 1,147 species and use these values as the global mean parameters. This gives us a single global calibration that can be transferred to a new dataset without using labels from the target deployment.

\textbf{Mean Likelihood Platt Scaling Plus (MLPS+)}
We observe significant variation in fitted calibration parameters between species (~\autoref{fig:coef_variation}). To attempt to explain this variability, we train unregularized linear regression models to predict the parameters $(m_s, b_s)$ from dataset-agnostic features: the mean-centered log of the species training count, and a measure of species acoustic similarity (\ref{app:acoustic_similarity}).

Using these species-level features, the regression predicts species-specific calibration parameters $\hat{m}_s$ and $\hat{b}_s$, which are then applied through Platt scaling. This allows the calibration function to vary between species while still requiring no labels from the target dataset.

\textbf{Oracular Platt Scaling (OPS)} provides an empirical upper bound on calibration performance by applying Platt scaling directly on the dataset's ground-truth labels. Unlike MLPS and MLPS+, OPS requires labels from the target dataset and therefore serves only as an empirical reference.

\subsection{Fusion Formulations}
\label{sec:fusion}

We establish distinct fusion algorithms $\phi(\cdot)$ for calculating the final joint probability $P(s|A,G)$. Previous work assumes categorical classification, and adopts a multiplicative fusion strategy (NBM, below). We introduce a binary joint fusion formula:

\begin{theorem}
Assuming independent audio and geographic evidence:
\begin{equation}
P(s \mid A,G)
=
\frac{
    P(s \mid A)
    P(s \mid G)
    P(\bar{s})
}{
    P(s \mid A)
    P(s \mid G)
    P(\bar{s})
    +
    P(\bar{s} \mid A)
    P(\bar{s} \mid G)
    P(s)
}.
\label{eq:foundational_bayes}
\end{equation}
\end{theorem}

The full derivation of Equation~\ref{eq:foundational_bayes} is provided in Appendix~\ref{app:derivation}. From this foundational equation, we establish two fusion approaches depending on the treatment of the prior probability $P(s)$.

\textbf{Naive Bayes Scaled (NBS)}
adopts Equation~\ref{eq:foundational_bayes} with the uninformative prior $P(s) = \frac{1}{2}$. Setting $\epsilon=10^{-9}$ for numerical stability, this yields the simplified fusion formula:
\begin{equation}
    P_{\mathrm{fused}}
    =
    \frac{
        P(s \mid A)P(s \mid G)
    }{
        \max\left(
            P(s \mid A)P(s \mid G)
            +
            \left(P(\bar{s} \mid A)\right)
            \left(P(\bar{s} \mid G)\right),
            \epsilon
        \right)
    }.
    \label{eq:nbs}
\end{equation}

\textbf{Naive Bayes Scaled with Prior (NBSP)} applies Equation~\ref{eq:foundational_bayes} with a geographical species prior, obtained as the mean geoprior for the species across all sites in the dataset:
\begin{equation}
    P(s) = \frac{1}{N} \sum_{i=1}^{N} P(s \mid G_i)
    \label{eq:prior_calculation}
\end{equation}

\textbf{Naive Bayes Multiplicative Fusion (NBM)} is derived from the joint-evidence normalization of \citep{macaodha2019presence} in a categorical classification scenario, assuming an constant class prior of $1/C$. This reduces to simple multiplication of the classifier and geoprior probabilities:
\begin{equation}
    P_{\mathrm{fused}}
    =
    P(s \mid A)
    \times
    P(s \mid G).
    \label{eq:nbm}
\end{equation}

\subsection{Evaluation Protocol and Metrics}
\label{sec:metrics}

We select a set of metrics which provide useful signal for multi-label, multi-class data under extreme label imbalance, which is common in bioacoustic datasets.

\textbf{ROC-AUC and Top-1} are used for evaluating model discrimination. ROC-AUC is computed independently for each species and then macro-averaged. Top-1 accuracy is computed by checking, for each clip with at least one species present, whether the highest scoring species is present in the clip.

\textbf{Unweighted Expected Calibration Error (uECE)}: Standard ECE measures calibration quality by partitioning predictions into probability bins, and computing a per-bin calibration error as the absolute difference between accuracy and confidence. The ECE is the bin-weighted average of these calibration errors. Because each species is typically quite sparse in bioacoustic datasets, standard ECE is dominated by the lowest bin. Instead, we give each populated probability bin equal weight:
\begin{equation}
    \mathrm{uECE}
    =
    \frac{1}{|\mathcal{B}|}
    \sum_{B_m \in \mathcal{B}}
    \left|
        \operatorname{acc}(B_m)
        -
        \operatorname{conf}(B_m)
    \right|,
    \label{eq:uece}
\end{equation}

where $\mathcal{B}$ is the set of populated bins. Lower ECE values indicate better calibration.

\textbf{Balanced KL Divergence}: Global calibration metrics can be deceptive; a method may appear well-calibrated on average while exhibiting substantial underconfidence, overconfidence, or variation in reliability for individual species. To measure per-species calibration, we compute Kullback-Leibler (KL) divergence between the predicted distribution and the oracular Platt scaling. To handle the sparsity of any particular species in the dataset, we apply a class balanced weighting to the metric.

Let $y_i \in \{0, 1\}$ be the empirical ground truth for window $i$, with $N_{pos}$ and $N_{neg}$ representing the total number of true positives and true negatives in the dataset. Let $\tilde{p}_i$ and $\tilde{q}_i$ be the predicted probabilities from the OPS reference and the evaluated fusion method, respectively. To prevent numerical instability, both probabilities are strictly clipped to the range $[\epsilon, 1 - \epsilon]$, where $\epsilon = 10^{-15}$.

To ensure the metric is class-balanced, we apply a sample weight $w_i$, assigning $w_i = 1$ for positive examples and $w_i = N_{pos}/N_{neg}$ for negative examples.

The balanced KL divergence ($bKL$) is the weighted sum of the pointwise Bernoulli KL divergences:
\begin{equation}
    bKL = \frac{\sum_{i=1}^{N} w_i \left[ \tilde{p}_i \ln\left(\frac{\tilde{p}_i}{\tilde{q}_i}\right) + (1 - \tilde{p}_i) \ln\left(\frac{1 - \tilde{p}_i}{1 - \tilde{q}_i}\right) \right]}{\sum_{i=1}^{N} w_i}
    \label{eq:balanced_kl_final}
\end{equation}

This balanced metric ensures that the divergence signal of rare positive calls is preserved, with a lower final score indicating that the evaluated method's outputs remain closer to OPS calibration.

\textbf{Reliability Diagrams} visualize the uECE metric by plotting the per-bin confidence (on the x-axis) against the per-bin accuracy (on the y-axis). The diagonal indicates perfect calibration, while points off the diagonal signal under-confidence (above diagonal) or over-confidence (below diagonal).

\section{Experiments}
We organized our experiments to separately measure the effects of pre-calibration, geoprior fusion, and acoustic model choice. Across experiments, Doohan (ADS) was our development dataset and Panwar (PDS) was our final evaluation dataset. However we report both results together so that changes observed on one dataset can be compared with a second evaluation setting.

\textbf{A.1 Pre-calibration of acoustic model outputs}:
We first evaluate calibration alone. We compare RAW, MLPS, and OPS for both Perch and BirdNET on ADS and PDS. This allows us to identify which kinds of pre-calibration are most effective.

\textbf{A.2 Explaining calibration variability}: We then ask whether the variation in WABADs species-specific Platt coefficients can itself be predicted without target-dataset labels. We compare species-level feature set using held-out species linear regression and measure how much variation in the coefficient they explain. Based on this analysis MLPS+ predicts species-specific Platt parameters from transferable features including pretraining count and acoustic species similarity. We compare MLPS+ with the simpler MLPS baseline on both evaluation datasets. MLPS+ is evaluated only for Perch because the required pretraining-count information is unavailable for BirdNET. 

This experiment isolates whether calibration learned from a separate dataset (WABAD) can improve acoustic probabilities before fusion, and whether that behavior is consistent across both datasets.

\textbf{B. Geoprior fusion of calibrated predictions}:
We next compare how different pre-calibration functions $f(\cdot)$ and fusion operations $\phi(\cdot)$ interact. We compare MLPS and MLPS+ alone with each of the three fusion methods: NBS, NBSP, and NBM. We also compare fusion with raw acoustic outputs to test whether geographic information alone can compensate for uncalibrated acoustic scores. The same comparisons are evaluated on both datasets.

\textbf{C. Acoustic-model ablation}:
Finally, we test whether the behavior of the fusion methods depends on the acoustic model used to provide $P(s|a)$. We replace Perch with BirdNET and repeat the geoprior fusion comparison using MLPS pre-calibration and the same NBS, NBSP, and NBM fusion formulations. Results are again evaluated on both ADS and PDS, allowing us to separate effects that are consistent across acoustic models from ones that may change depending on the model.

\section{Results}

\begin{table}[ht]
\centering
\caption{Pre-calibration performance of Perch and BirdNET. 
Reporting Australia (A) and Panwar (P). Best scores are bold, oracle scores are italic.}
\label{tab:precalibration}
\makebox[\textwidth][c]{%
\begin{tabular}{llcccccccc}
\toprule
Audio model &
Pre-cal. &
\multicolumn{2}{c}{ROC-AUC $\uparrow$} &
\multicolumn{2}{c}{Top-1 $\uparrow$} &
\multicolumn{2}{c}{uECE $\downarrow$} &
\multicolumn{2}{c}{bKL $\downarrow$} \\
\cmidrule(lr){3-4} \cmidrule(lr){5-6} \cmidrule(lr){7-8} \cmidrule(lr){9-10}
& & A & P & A & P & A & P & A & P \\
\midrule

Perch 2.0
    & None  & 0.927 & 0.930 & 0.778 & \textbf{0.848} & 0.550 & 0.614 & 2.144 & 4.232 \\
    & MLPS  & 0.927 & 0.930 & 0.778 & \textbf{0.848} & \textbf{0.315} & \textbf{0.216} & 0.255 & \textbf{0.063} \\
    & MLPS+ & 0.927 & 0.930 & \textbf{0.785} & 0.832 & 0.327 & 0.318 & \textbf{0.223} & 0.161 \\
    & \textit{OPS}   & \textit{0.930} & \textit{0.834} & \textit{0.849} & \textit{0.851} & \textit{0.012} & \textit{0.019} & \textit{0} & \textit{0} \\

\addlinespace

BirdNET v2.4
    & None  & 0.897 & 0.853 & 0.674 & 0.678 & 0.284 & 0.169 & 0.105 & 0.047 \\
    & MLPS  & 0.897 & 0.853 & 0.674 & 0.678 & \textbf{0.225} & \textbf{0.078} & \textbf{0.074} & \textbf{0.043} \\
    & \textit{OPS}   & \textit{0.897} & \textit{0.853} & \textit{0.753} & \textit{0.767} & \textit{0.053} & \textit{0.051} & \textit{0} & \textit{0} \\

\bottomrule
\end{tabular}
} % makebox
\end{table}

\textbf{A.1:  Pre-calibration for acoustic model outputs}

On Doohan, uncalibrated Perch achieves an ROC-AUC of 0.927 and Top-1 accuracy of 0.778, but suffers from an uECE of 0.550 and a KL divergence of 2.144, worse than a random-chance baseline (bKL 0.942). MLPS global calibration preserves acoustic discrimination while drastically reducing calibration error, reducing uECE to 0.315 and KL divergence to 0.255 on Doohan, with even larger reductions on PANWAR (~\autoref{tab:precalibration}). MLPS similarly improves BirdNET calibration, demonstrating that this method generalizes across both architectures and distinct ecosystems.

Oracular Platt Scaling (OPS), fit directly on target-deployment labels, provides an empirical reference. On PANWAR, however, OPS suffers a drop in macro ROC-AUC. This instability arises because 26\% of PANWAR species have less than 20 positive annotations; unconstrained logistic fits on these rare taxa often produce negative slopes that invert discriminative ranking. In contrast, global MLPS avoids this long-tail pathology by sharing robust calibration priors across all taxa.

\textbf{A.2 Explaining calibration variability}:
Although global MLPS substantially improved calibration, we wanted to understand variation among species. Variation in both the fitted slopes and intercepts can be seen in ~\autoref{fig:coef_variation}. We found that dataset-specific features like deployment and call density offered the highest joined explanation of variance. However, the dataset-agnostic features (pre-training count, and species similarity) explained a portion of coefficient variance on WABAD without requiring  dataset-specific labels (Table ~\ref{tab:mlpsplus-features}).

\begin{figure}[t]
    \centering
    \vspace{0pt} % 2. INVISIBLE ANCHOR FOR TOP ALIGNMENT
    \centering
    \includegraphics[width=0.6\linewidth]{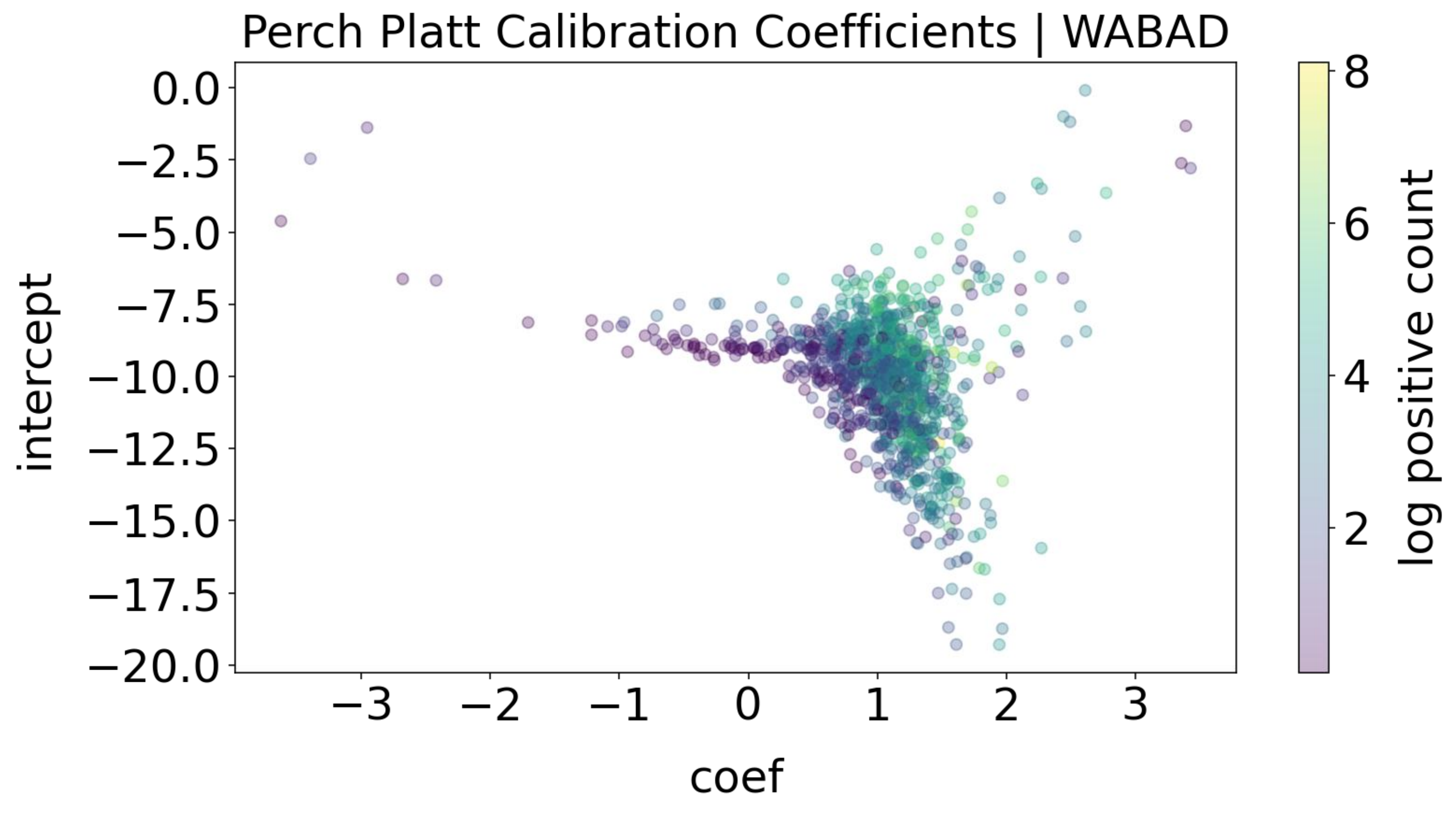}
    \caption{Variation in slope and intercept for Perch calibration on the WABAD dataset.}
    \label{fig:coef_variation}
\end{figure}

\begin{table}
    \centering
    \noindent % 1. PREVENTS MARGIN OVERFLOW
    
    % --- ROW 1: CONTENT (Tops Aligned) ---
        \small % 3. ICLR-COMPLIANT FONT SIZING
        \setlength{\tabcolsep}{3pt} % 4. REDUCES COLUMN PADDING TO FIT
        \begin{tabular}{lcc}
            \toprule
            Features & MSE $\downarrow$ & $R^2$ $\uparrow$ \\
            \midrule
            Log pretrain examples & 1.82 & 0.17 \\
            WABAD deployment & 2.03 & 0.15 \\
            WABAD call density & 2.66 & 0.09 \\
            Species similarity & 2.43 & 0.05 \\
            Taxonomic Family & 2.59 & -0.04 \\
            Taxonomic Order & 2.70 & -0.01 \\
             \hline n-pretrain, similarity & 1.80 & 0.18 \\
            n-pretrain, deployment & 1.63 & 0.23 \\
            n-pretrain, deployment, call density & 1.55 & 0.35 \\
            \bottomrule
        \end{tabular}
    
    \vspace{2mm} % Space between charts and captions
    
    \caption{Prediction of transfer-calibration coefficients from dataset and model features. Lower MSE and higher $R^2$ indicate better prediction.}
    \label{tab:mlpsplus-features}
    
\end{table}

Predicting species-specific calibration parameters provided mixed results compared to the simpler MLPS calibration. On the Australia dataset, MLPS+ preserved per-class ROC-AUC and increased top-1 accuracy. Balanced KL divergence also decreased compared to MLPS, indicating that the species-level calibrated outputs for MLPS+ were closer to the OPS reference. However uECE increased slightly. MLPS+ therefore improves some aspects of the predictions but does not dominate MLPS across evaluation criteria.

MLPS+ gains did not transfer to Panwar. MLPS achieved a uECE of 0.216 and KL divergence of 0.063, whereas MLPS+ increased these errors. Top-1 accuracy also decreased (~\autoref{tab:precalibration}).

\textbf{B: Fusion with calibrated Perch Predictions}:
The lowest uECE when fusing with uncalibrated perch scores was 0.453. Geoprior fusion alone did not compensate for poor acoustic model calibration, showing that pre-calibration is a critical step before fusion. 

Each fusion and pre-calibration methods produced tradeoffs on metrics. NBS most aggressively reduced uECE but could substantially reduce Top-1 Accuracy. NBM consistently produced the strongest discrimination but increased calibration error. For each fusion method, NBSP often resulted in best balance between low uECE scores and low KL divergence while still maintaining gains in ROC-AUC and Top 1 (~\autoref{tab:fusion-precal-paired}). It achieved the lowest balanced KL divergence among the three fusion methods in all pre-calibration settings, preserving substantially more Top-1 Accuracy than NBS and remaining considerably better calibrated than NBM. On a Reliability diagram, as seen in \autoref{fig:MLPS_Reliability}, MLPS with NBSP fusions did not always hug the exact 45$^\circ$ but did move closer to the diagonal. While on PANWAR, MLPS with NBSP hugs the 45$^\circ$ closely we see in its metrics that it saw a drop in Top-1 and an increase in balanced KL Divergence.

\begin{figure}[t]
    \centering
    \includegraphics[width=.98\linewidth]{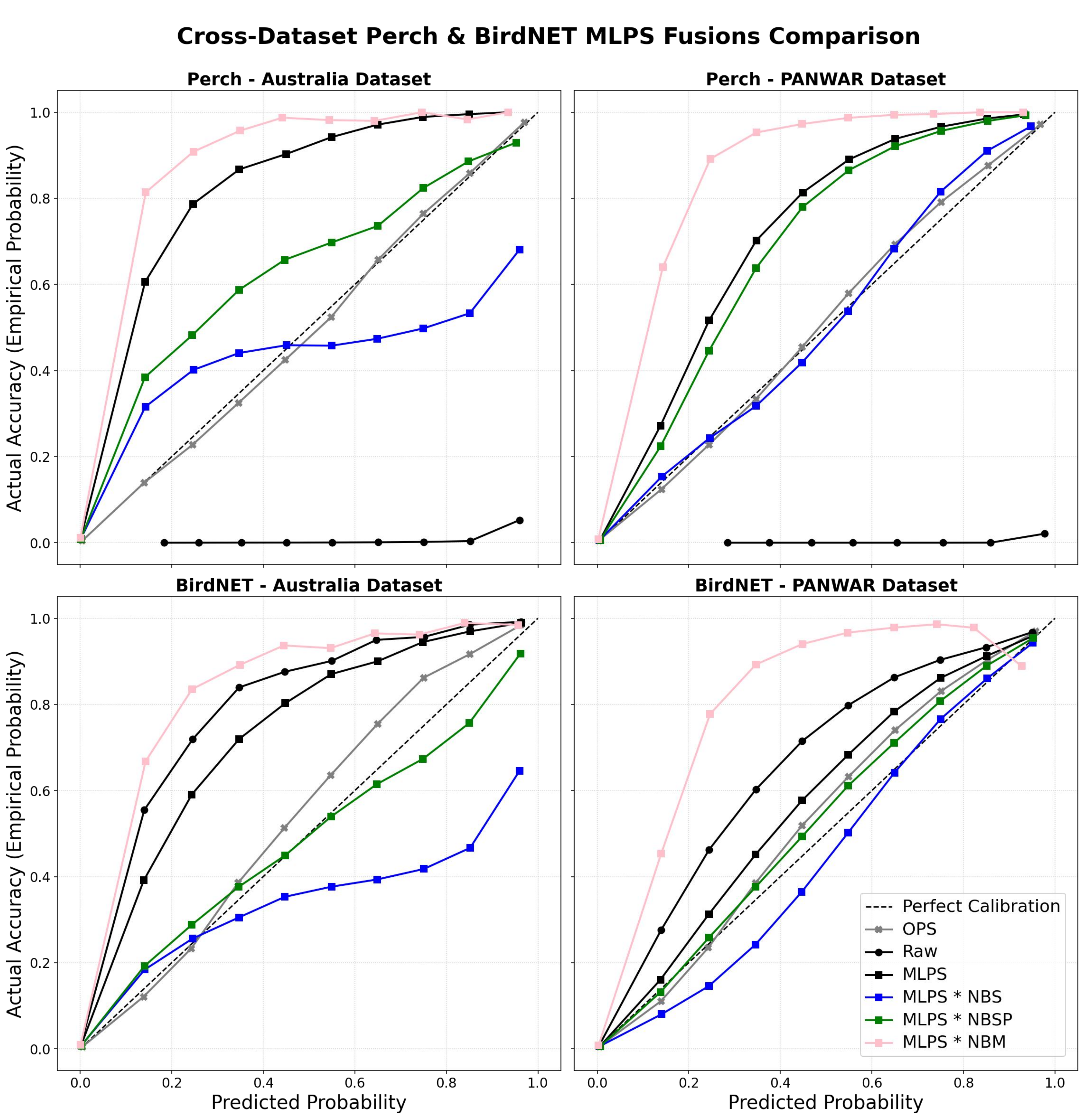}
    \caption{Cross-Dataset Reliability Diagrams on Australia (left) and PANWAR (right) for Perch (top) and BirdNET (bottom) across geoprior fusion methods.}
    \label{fig:MLPS_Reliability}
\end{figure}

\textbf{C:  Audio Model Ablation}: 
To understand the robustness of the NBSP and MLPS approach we tested this fusion and pre-calibration method with a new acoustic model as the audio source. The relative behavior of the fusion method remains consistent when BirdNET replaces Perch as the acoustic model. On Doohan, MLPS improved calibration relative to its raw outputs reducing uECE slightly and improving KL divergence from 0.105 to 0.074. Fusion reproduced the trade-offs observed with Perch: NBM produced the strongest gains in discrimination, but worsened calibration, while NBS substantially reduced Top-1 accuracy. NBSP provided the strongest calibration and discrimination balance. It reduced uECE to 0.039 and bKL to 0.064 while increasing ROC-AUC (see ~\autoref{tab:fusion-precal-paired}). These results suggest that the relative behavior of the fusion methods generalize across acoustic foundation models.

In the Reliability diagram with BirdNET as the audio model (\autoref{fig:MLPS_Reliability}), MLPS pre-calibration with NBSP fusion hugs close to the $45^\circ$ line. While NBM, which did increase discrimination, clearly suffers from systemic under confidence and raises well above the $45^\circ$ line.

\begin{table}[t!]
\centering
\caption{Fusion performance with various pre-calibration methods, reported on the Doohan~(A) and Panwar~(P) datasets.}
\label{tab:fusion-precal-paired}
    
\makebox[\textwidth][c]{%
    \begin{tabular}{llcccccccc}
    \toprule
    Audio model &
    Fusion &
    \multicolumn{2}{c}{ROC-AUC $\uparrow$} &
    \multicolumn{2}{c}{Top-1 $\uparrow$} &
    \multicolumn{2}{c}{uECE $\downarrow$} &
    \multicolumn{2}{c}{bKL $\downarrow$} \\
    \cmidrule(lr){3-4} \cmidrule(lr){5-6} \cmidrule(lr){7-8} \cmidrule(lr){9-10}
    & & A & P & A & P & A & P & A & P \\
    \midrule
    \multicolumn{10}{l}{\textbf{No pre-calibration or Fusion}} \\
    \midrule
    Perch 2.0
        & none  & 0.927 & 0.930 & 0.778 & 0.848 & 0.550 & 0.614 & 2.144 & 4.232 \\
        
    BirdNET v2.4
        & none  & 0.900 & 0.853 & 0.674 & 0.678 & 0.284 & 0.160 & 0.105 & 0.047 \\
    
    \midrule
    
    \multicolumn{10}{l}{\textbf{No pre-calibration}} \\
    \midrule
    
    Perch 2.0
        & NBS  & 0.930 & 0.933 & 0.524 & 0.760 & 0.492 & 0.498 & 1.989 & 3.348 \\
        & NBSP & 0.930 & 0.933 & 0.700 & 0.851 & 0.494 & 0.514 & 2.359 & 4.268 \\
        & NBM  & 0.815 & 0.703 & 0.187 & 0.046 & 0.453 & 0.468 & 0.439 & 0.547 \\
    
    \midrule
    \multicolumn{10}{l}{\textbf{MLPS pre-calibration}} \\
    \midrule
    
    Perch 2.0
        & NBS  & 0.930 & 0.933 & 0.510 & 0.749 & 0.155 & \textbf{0.026} & 0.277 & 0.162 \\
        & NBSP & 0.930 & 0.933 & 0.690 & 0.850 & \textbf{0.130} & \textit{0.188} & \textbf{0.195} & \textbf{0.072} \\
        & NBM  & \textbf{0.937} & \textbf{0.934} & \textit{0.769} & \textit{0.853} & 0.373 & 0.356 & 0.390 & 0.283 \\
    
    \midrule
    
    BirdNET v2.4
        & NBS  & 0.912 & 0.860 & 0.462 & 0.676 & 0.166 & 0.043 & 0.100 & 0.085 \\
        & NBSP & 0.912 & 0.860 & 0.600 & 0.675 & 0.039 & 0.032 & 0.064 & 0.054 \\
        & NBM  & 0.913 & 0.860 & 0.675 & 0.727 & 0.326 & 0.308 & 0.159 & 0.146 \\
    
    \midrule
    \multicolumn{10}{l}{\textbf{MLPS+ pre-calibration}} \\
    \midrule
    
    Perch 2.0
        & NBS  & 0.930 & 0.933 & 0.510 & 0.833 & 0.155 & 0.301 & 0.277 & 0.285 \\
        & NBSP & 0.930 & 0.933 & 0.691 & 0.837 & \textit{0.131} & 0.300 & \textit{0.195} & \textit{0.156} \\
        & NBM  & \textbf{0.937} & \textbf{0.934} & \textbf{0.778} & \textbf{0.857} & 0.373 & 0.439 & 0.360 & 0.412 \\
    
    \midrule
    \multicolumn{10}{l}{\textbf{OPS pre-calibration}} \\
    \midrule
    
    Perch 2.0
        & NBS  & 0.929 & 0.841 & 0.663 & 0.778 & 0.253 & 0.196 & 0.059 & 0.105 \\
        & NBSP & 0.929 & 0.841 & 0.801 & 0.853 & 0.184 & 0.018 & 0.051 & 0.005 \\
        & NBM  & \textit{0.935} & 0.841 & \textit{0.779} & 0.812 & 0.214 & 0.202 & 0.107 & 0.232 \\
    \midrule
    \end{tabular}
}
\end{table}

\section{Discussion}

Our results demonstrate that bioacoustic models can often be calibrated without local labels by using parameters from global annotated datasets and subsequent fusion with geopriors. 

\textbf{Pre-calibration can be transferred across datasets:}
Global calibration parameters derived from the global dataset (WABAD) improved calibration for both Perch v2 and BirdNET on both the Doohan and Panwar datasets (MLPS). However, we did not see significantly better calibration when adjusting calibration to account for the amount of pretraining data for each species (MLPS+). This suggest that species-level variation in calibration parameters is partially predictable, but predicting that variation does not generalize uniformly across datasets. The simpler global MLPS provides a more consistent balance between calibration and discrimination, likely because it corrects a systemic pattern in the score distribution of an acoustic model without overfitting.

\textbf{GeoPrior Fusion Methods and Pre-calibration:}
Our experiments show that geographic evidence alone does not correct acoustic model calibration. Pre-calibration is essential to ensure the inputs to fusion are on a compatible scale. Once calibrated, the choice of fusion method leads to a tradeoff between ranking and reliability.

NBM allowed the prior to suppress implausible species and improve ranking. However, because NBM multiplies two bounded probability estimates, it can compress predictions toward zero: Even when both models are quite confident, the product leads to under-confidence. NBSP resolves this by using the full Bayes formulation, handling each species prediction as a binary classification task. By normalizing the joint evidence against $P(s)$, NBSP scales the probabilities back up. It successfully preserves the geographic filtering benefits of NBM while improving calibration.

Relying on a single metric to evaluate this calibration is deceptive. For example, a model that ignores audio entirely and simply predicts the dataset's base rate will achieve a near-perfect uECE and a flawless reliability diagram. Failure only becomes clear when checking ROC-AUC, Top-1, and balanced KL divergence. Because of this, fusion methods must be evaluated holistically. NBSP proves to be the strongest method because it maintains classification accuracy (ROC-AUC) and matches empirical distributions (KL divergence), rather than just minimizing uECE.

\textbf{Generalization and Ecological Implications:}
Pre-calibration and careful geoprior fusion generalize across Perch and BirdNET, and to both North American and Australian datasets.

We believe that our approach generalizes well across different acoustic models because it targets fundamental behaviors. MLPS captures and corrects baseline acoustic model behavior on a global scale, regardless of which acoustic model the method is being used on. By fusing this globally corrected acoustic signal with local geographic priors, we combine pre-calibrated acoustic evidence with a specific geographic environment. The resulting better calibrated scores could enable researchers to choose meaningful precision-based thresholds with less effort or to utilize the full score distribution for tasks like estimating call density. 

\textbf{Limitations and future work:}
While transferred calibration (MLPS) significantly improves calibration, it does not perfectly recreate empirical oracle calibration (OPS) and may weaken under domain shifts. Furthermore, unconstrained OPS can become unstable for rare species, suggesting that testing monotonic constraints is a necessary future improvement. Because our experiments relied on a single geoprior model, future work should evaluate alternatives; we expect that the quality and calibration of the geoprior model is also a factor in effectiveness of fusion. While all datasets in this work contain multiple sites, evaluation on more datasets would improve the confidence in our results: We were restricted by the unavailability of bioacoustic datasets with fine-grained geographic location data. Finally, the natural next step is to test whether these fused scores directly improve downstream quantitative tasks. Prior work, like that in ~\citet{navine2024thresholds}, demonstrates that ecological metrics, such as call density, can be directly estimated from the full distribution of classifier scores; testing our calibrated outputs with these methods will determine their effectiveness.

\section{Conclusion}

If bioacoustics is to monitor global biodiversity at scale, calibration must be treated as a central objective alongside classification accuracy. We demonstrate that the bottleneck of site-specific manual labeling may not be required to achieve this. By transferring global calibration priors and fusing them with continuous geopriors via NBSP, raw model outputs can be transformed into calibrated probabilities. This framework preserves the rich, continuous information that rigid thresholds discard, clearing a path for simple and reliable broad biodiversity monitoring.

% \subsubsection*{Author Contributions}

% If you'd like to, you may include  a section for author contributions as is done
% in many journals. This is optional and at the discretion of the authors.

% \subsubsection*{Acknowledgments}
% Use unnumbered third level headings for the acknowledgments. All
% acknowledgments, including those to funding agencies, go at the end of the paper.

\bibliographystyle{iclr2027_conference}
\bibliography{references}

@article{vanmerrienboer2025perch,
  author    = {van Merri{\"e}nboer, Bart and Dumoulin, Vincent and Hamer, Jenny and Simpson, Isabelle and Burns, Andrea and Harrell, Lauren and Denton, Tom},
  title     = {Perch 2.0: Advanced Multi-taxa Bioacoustics},
  journal   = {arXiv preprint arXiv:2508.00000},
  year      = {2025}
}

@article{kahl2021birdnet,
  author    = {Kahl, Stefan and Wood, Connor M. and Eibl, Maximilian and Klinck, Holger},
  title     = {BirdNET: A deep learning solution for avian diversity monitoring},
  journal   = {Ecological Informatics},
  volume    = {61},
  pages     = {101236},
  year      = {2021},
  doi       = {10.1016/j.ecoinf.2021.101236}
}

@article{august2022emerging,
  author    = {August, Tom A. and Harvey, Martin C. and Subedar, Zareen and et al.},
  title     = {Emerging technologies revolutionise insect ecology and monitoring},
  journal   = {Philosophical Transactions of the Royal Society B: Biological Sciences},
  volume    = {377},
  number    = {1862},
  pages     = {20210089},
  year      = {2022},
  doi       = {10.1098/rstb.2021.0089}
}

@article{wood2024guidelines,
  author    = {Wood, Connor M. and Kahl, Stefan},
  title     = {Guidelines for appropriate use of BirdNET scores and other detector outputs},
  journal   = {Journal of Ornithology},
  volume    = {165},
  number    = {3},
  pages     = {777--782},
  year      = {2024},
  doi       = {10.1007/s10336-024-02144-5}
}

@article{schwinger2025uncertainty,
  author    = {Raphael Schwinger and Ben McEwen and Vincent S. Kather and René Heinrich and Lukas Rauch and Sven Tomforde},
  title     = {Uncertainty Calibration of Multi-Label Bird Sound Classifiers},
  journal   = {arXiv preprint arXiv:2511.08261},
  year      = {2025}
}

@article{jeantet2023,
title = {Improving deep learning acoustic classifiers with contextual information for wildlife monitoring},
journal = {Ecological Informatics},
volume = {77},
pages = {102256},
year = {2023},
issn = {1574-9541},
doi = {https://doi.org/10.1016/j.ecoinf.2023.102256},
url = {https://www.sciencedirect.com/science/article/pii/S1574954123002856},
author = {Lorène Jeantet and Emmanuel Dufourq}
}

@article{a2o,
author = {Roe, Paul and Eichinski, Philip and Fuller, Richard A. and McDonald, Paul G. and Schwarzkopf, Lin and Towsey, Michael and Truskinger, Anthony and Tucker, David and Watson, David M.},
title = {The Australian Acoustic Observatory},
journal = {Methods in Ecology and Evolution},
volume = {12},
number = {10},
pages = {1802-1808},
doi = {https://doi.org/10.1111/2041-210X.13660},
url = {https://besjournals.onlinelibrary.wiley.com/doi/abs/10.1111/2041-210X.13660},
eprint = {https://besjournals.onlinelibrary.wiley.com/doi/pdf/10.1111/2041-210X.13660},
year = {2021}
}

@article{knight2017recommendations,
  author    = {Knight, Elly C. and Hannah, Kevin C. and Foley, Gabriel J. and Scott, Christopher D. and Brigham, R. Mark and Bayne, Erin M.},
  title     = {Recommendations for acoustic recognizer performance assessment with application to five common automated signal recognition programs},
  journal   = {Avian Conservation and Ecology},
  volume    = {12},
  number    = {2},
  pages     = {14},
  year      = {2017},
  doi       = {10.5751/ACE-01114-120214}
}

@article{cole2022occupancy,
    author = {Cole, Jerry S and Michel, Nicole L and Emerson, Shane A and Siegel, Rodney B},
    title = {Automated bird sound classifications of long-duration recordings produce occupancy model outputs similar to manually annotated data},
    journal = {Ornithological Applications},
    volume = {124},
    number = {2},
    pages = {duac003},
    year = {2022},
    month = {05},
    issn = {0010-5422},
    doi = {10.1093/ornithapp/duac003},
    url = {https://doi.org/10.1093/ornithapp/duac003},
    eprint = {https://academic.oup.com/condor/article-pdf/124/2/duac003/43599793/duac003.pdf},
}

@inproceedings{guo2017calibration,
  author    = {Guo, Chuan and Pleiss, Geoff and Sun, Yu and Weinberger, Kilian Q.},
  title     = {On Calibration of Modern Neural Networks},
  booktitle = {Proceedings of the 34th International Conference on Machine Learning (ICML)},
  volume    = {70},
  pages     = {1321--1330},
  year      = {2017}
}

@inproceedings{ye2022uncertainty,
  author    = {Ye, Tong and Si, Shijing and Wang, Jianzong and Cheng, Ning and Xiao, Jing},
  title     = {Uncertainty Calibration for Deep Audio Classifiers},
  booktitle = {Interspeech 2022},
  pages     = {1556--1560},
  year      = {2022},
  doi       = {10.21437/Interspeech.2022-11384}
}

@article{navine2024thresholds,
  author    = {Navine, Amanda K. and Denton, Tom and Weldy, Matthew J. and Hart, Patrick J.},
  title     = {All thresholds barred: direct estimation of call density in bioacoustic data},
  journal   = {Frontiers in Bird Science},
  volume    = {3},
  pages     = {1380636},
  year      = {2024},
  doi       = {10.3389/fbirs.2024.1380636}
}

@article{sugai2019terrestrial,
  author    = {Sugai, L{\'a}ris S. M. and Silva, Thiago S. F. and Ribeiro, Jos{\'e} W. and Llusia, Diego},
  title     = {Terrestrial Passive Acoustic Monitoring: Review and Perspectives},
  journal   = {BioScience},
  volume    = {69},
  number    = {1},
  pages     = {15--25},
  year      = {2019},
  doi       = {10.1093/biosci/biy147}
}

@article{barre2019accounting,
  author    = {Barr{\'e}, K{\'e}vin and Le Viol, Isabelle and Bas, Yves},
  title     = {Accounting for automated identification errors in acoustic surveys},
  journal   = {Methods in Ecology and Evolution},
  volume    = {10},
  number    = {8},
  pages     = {1171--1188},
  year      = {2019},
  doi       = {10.1111/2041-210X.13219}
}

@article{elith2009species,
  author    = {Elith, Jane and Leathwick, John R.},
  title     = {Species Distribution Models: Ecological Explanation and Prediction Across Space and Time},
  journal   = {Annual Review of Ecology, Evolution, and Systematics},
  volume    = {40},
  number    = {1},
  pages     = {677--697},
  year      = {2009},
  doi       = {10.1146/annurev.ecolsys.110308.120159}
}

@inproceedings{macaodha2019presence,
  author    = {Mac Aodha, Oisin and Cole, Elijah and Perona, Pietro},
  title     = {Presence-Only Geographical Priors for Fine-Grained Image Classification},
  booktitle = {Proceedings of the IEEE/CVF Conference on Computer Vision and Pattern Recognition (CVPR)},
  pages     = {9596--9606},
  year      = {2019}
}

@inproceedings{platt1999probabilistic,
  author    = {Platt, John C.},
  title     = {Probabilistic Outputs for Support Vector Machines and Comparisons to Regularized Likelihood Methods},
  booktitle = {Advances in Large Margin Classifiers},
  editor    = {Smola, Alexander J. and Bartlett, Peter L. and Sch{\"o}lkopf, Bernhard and Schuurmans, Dale},
  pages     = {61--74},
  year      = {1999},
  publisher = {MIT Press}
}

@article{perezgranados2026wabad,
  title={{WABAD}: A World Annotated Bird Acoustic Dataset for Passive Acoustic Monitoring},
  author={P{\'e}rez-Granados, Cristian and Morant, Jon and Darras, Kevin F.~A. and Mar{\'i}n-G{\'o}mez, Oscar H. and Mendoza, Irene and Mu{\~n}oz-Mohedano, Miguel A. and Santamar{\'i}a-Garc{\'i}a, Eduardo and Bastianelli, Giulia and M{\'a}rquez-Rodr{\'i}guez, Alba and Budka, Micha{\l} and Panwar, Pooja and Weed, Aaron S. and Sebasti{\'a}n-Gonz{\'a}lez, Esther},
  journal={Ecology},
  volume={107},
  number={2},
  pages={e70317},
  year={2026},
  doi={10.1002/ecy.70317}
}

@article{panwar2025largescale,
  title={Large-scale, Stratified, Fully Annotated Acoustic Forest Soundscape Dataset of Avian Vocalizations from Eastern North America},
  author={Panwar, Pooja and Cummings, Wyatt J. and Martinson, Sharon Jean and Weed, Aaron S. and Symes, Laurel B. and ter Hofstede, Hannah M.},
  journal={Scientific Data},
  year={2025},
  doi={10.5281/zenodo.18041381}
}

\appendix
\section{Appendix}

\subsection{Derivation of the Joint Posterior Probability}
\label{app:derivation}

\begin{proof}
Applying Bayes' theorem and expanding the denominator using the law of total probability gives

\begin{equation}
    P(s \mid A,G)
    =
    \frac{
        P(A,G \mid s)P(s)
    }{
        P(A,G \mid s)P(s)
        +
        P(A,G \mid \bar{s})P(\bar{s})
    }.
    \label{eq:joint_posterior}
\end{equation}

To integrate the two model outputs, we assume conditional independence between the audio and geographic evidence given species presence or absence:

\begin{equation}
    P(A,G \mid s)
    =
    P(A \mid s)P(G \mid s).
    \label{eq:conditional_independence}
\end{equation}

Applying Bayes' theorem to the individual likelihoods gives

\begin{equation}
    P(A \mid s)
    =
    \frac{P(s \mid A)P(A)}{P(s)}
\end{equation}
and
\begin{equation}
    P(G \mid s)
    =
    \frac{P(s \mid G)P(G)}{P(s)}.
\end{equation}

Therefore, the joint likelihood becomes

\begin{equation}
    P(A,G \mid s)
    =
    \frac{
        P(s \mid A)P(s \mid G)P(A)P(G)
    }{
        P(s)^2
    }.
\end{equation}

Substituting the corresponding expressions for both presence and absence into Equation~\ref{eq:joint_posterior}, the marginal evidence terms $P(A)P(G)$ cancel out, resulting in the final foundational expression:

\begin{equation}
    P(s \mid A,G)
    =
    \frac{
        \dfrac{P(s \mid A)P(s \mid G)}{P(s)}
    }{
        \dfrac{P(s \mid A)P(s \mid G)}{P(s)}
        +
        \dfrac{P(\bar{s} \mid A)P(\bar{s} \mid G)}{P(\bar{s})}
    }.
\end{equation}
\end{proof}

\subsection{Species Acoustic Similarity}
\label{app:acoustic_similarity}

To measure how likely a given species is to be confused with some other species in the WABAD dataset, we calculate a species co-occurrence similarity score relative to each other species, and then take the maximum score.

To calculate the co-occurrence similarity, we use Perch v2 to extract raw acoustic logits across all temporal windows in the target deployment and binarize these predictions using a (relatively low) logit threshold of $7.0$. We then compute a pairwise Jaccard similarity matrix across all species detections to measure how frequently target species co-occur. Finally, we isolate the single highest similarity score each species shares with any other species (excluding its correlation with itself) and apply a natural log transformation.

\section{AI disclosure}
In this work, we used generative AI tools to write code used in our data processing and fusion experiments. We have not used generative AI tools for generating synthetic data, formulating theoretical models, proving mathematical claims, or interpreting results, and the rest of the required disclosure tasks are not applicable. Additionally, we used AI tools to edit the paper for readability, format LaTeX, and identify relevant literature. We have reviewed all AI-assisted work: LLM-generated code was tested for correctness by the authors, and AI-edited text was manually verified for accuracy. We take responsibility for the final content of this work, including text edited with the aid of generative AI.

\appendix

\end{document}